\documentclass[letterpaper, 10 pt, conference]{ieeeconf}  % Comment this line out
\IEEEoverridecommandlockouts                              % This command is only
\usepackage[lined,ruled,vlined]{algorithm2e}

\usepackage{mathtools}

\usepackage{bm}
\usepackage{srcltx}
\usepackage{cite}    
\usepackage{verbatim}
\usepackage{float}
\usepackage{color}
\usepackage{acronym}\usepackage{subcaption}
\usepackage{caption}
\usepackage{subcaption}

\usepackage{mathtools}

\usepackage[lined,ruled,vlined]{algorithm2e}

\usepackage{amssymb,latexsym,amsfonts,amsmath,amscd} 
\usepackage{paralist}
\usepackage[left=20mm, right=20mm, top=20mm,
  bottom=20mm]{geometry}
 \newtheorem{definition}{Definition}
 \newtheorem{problem}{Problem}
\newtheorem{assumption}{Assumption}
\newtheorem{remark}{Remark}
\usepackage{tikz}
\usetikzlibrary{shapes,snakes}

\usepackage{placeins}
\usetikzlibrary{arrows,fit,shapes,automata}
\usetikzlibrary{positioning,fit,calc,shapes}
\usetikzlibrary{decorations.fractals}
\usetikzlibrary{decorations.markings}

\newtheorem{theorem}{Theorem}
 
\newtheorem{lemma}{Lemma}
\newtheorem{example}{Example}

\acrodef{mitl}[MITL]{metric interval temporal logic}
\acrodef{sde}[SDE]{stochastic differential equation}
\acrodef{dfta}[DFTA]{deterministic finite-state timed automaton}
\acrodef{mdp}[MDP]{Markov decision process}

\newcommand{\calA}{\mathcal{A}}
\newcommand{\calU}{\mathcal{U}}
\newcommand{\calT}{\mathcal{T}}

\newcommand{\bbR}{\mathbb{R}}
\newcommand{\bbZ}{\mathbb{Z}_{\ge 0}}
\newcommand{\Occ}{\mathsf{Occ}}

\newcommand{\calV}{\mathcal{V}}
\newcommand{\calF}{\mathcal{F}}
\newcommand{\calM}{\mathcal{M}}

\newcommand{\calAP}{\mathcal{AP}}

  \newcommand{\sink}{\mathsf{sink}}
  \newcommand{\truev}{\top}

  \newcommand{\init}{\mathsf{Init}}

 \newcommand{\nat}{\mathbb{N}}
\providecommand{\abs}[1]{\left|#1\right|}
\providecommand{\norm}[1]{\lVert#1\rVert}

\title{\LARGE \bf Computational methods for stochastic control with
  metric interval temporal logic
  specifications % <-this % stops a space
  \author{Jie Fu and Ufuk Topcu}
\thanks{J. Fu and U. Topcu are with the Department of Electrical and
  Systems Engineering, University of Pennsylvania, Philadelphia, PA
  19104, USA  {\tt\small jief, utopcu@seas.upenn.edu}.}}

\begin{document}
\maketitle
\thispagestyle{empty}
\pagestyle{empty}

%%%%%%%%%%%%%%%%%%%%%%%%%%%%%%%%%%%%%%%%%%%%%%%%%%%%%%%%%%%%%%%%%%%%%%%%%%%%%%%%
\begin{abstract}
  This paper studies an optimal control problem for continuous-time
  stochastic systems subject to reachability objectives specified in a
  subclass of metric interval temporal logic specifications, a
  temporal logic with real-time constraints. We propose a
  probabilistic method for synthesizing an optimal control policy that
  maximizes the probability of satisfying a specification based on a
  discrete approximation of the underlying stochastic system.  First,
  we show that the original problem can be formulated as a stochastic
  optimal control problem in a state space augmented with finite
  memory and states of some clock variables. Second, we present a
  numerical method for computing an optimal policy with which the
  given specification is satisfied with the maximal probability in
  point-based semantics in the discrete approximation of the
  underlying system. We show that the policy obtained in the discrete
  approximation converges to the optimal one for satisfying the
  specification in the continuous or dense-time semantics as the
  discretization becomes finer in both state and time.  Finally, we
  illustrate our approach with a robotic motion planning example.
\end{abstract}
\section{Introduction}
\label{sec:intro}
Stochastic optimal control is an important research area for analysis
and control design for continuous-time dynamical systems that operate
in the presence of uncertainty. % In the last decades, theories
% and tools have been developed for stochastic control with respect to
% control objectives such as optimal stopping problems, discounted
% cost problems with a target set, average cost per unit time and
% finite-time control problems ~\cite{}.  Besides tranditional control
% objectives, there is a rapid growing interest in high-level formal
% specifications with real-time constraints for autonomous and
% semiautonomous systems. For example, a robot may be required to
% visit some target region within a bounded time interval, while
% avoiding unsafe region before its arrivable to the target one.
However, existing stochastic control methods cannot be readily applied
to handle complex temporal logic specifications with real-time
constraints, which are of growing interest to the design of autonomous
and semiautonous systems
\cite{pnueli1992temporal,koymans1990specifying,fainekos2009temporal,wongpiromsarn2012receding}. In
this paper, we propose a numerical method for stochastic optimal
control with respect to a subclass of metric temporal logic
specifications. Particularly, given a specification encoding desirable
properties of a continuous-time stochastic system, the task is to
synthesize a control policy such that if the system implements the
policy, then the probability of a path satisfying the formula is
maximized.
 % With uncertainty in its dynamics and its environment, the robot
% aims to maximize the probability for accomplishing its assigned task.

% This paper proposes a numerical method for stochastic optimal
% control problems with respect to a subclass of metric temporal logic
% (MTL) specifications.

Metric temporal logic (MTL) is one of many real-time logics that not
only express the relative temporal ordering of events as linear
temporal logic (LTL), but also the duration between these events. For
example, a surveillance task of a mobile robot, infinitely revisiting
region 1 and 2, can be expressed in LTL. But tasks with quantitative
timing constraints, for instance, visiting region 2 within 5 minutes
after visiting region 1, require the expressive power of MTL. For
system specifications in LTL and its untimed variants, methods have
been developed for quantitative verification of discrete-time
stochastic hybrid systems
\cite{abate2010approximate,abate2011quantitative}, control design of
continuous-time and discrete-time linear stochastic systems
\cite{lahijanian2009probabilistic,Maria2015}.  For MTL and its
variants, a specification-guided testing framework is proposed in
\cite{Abbas2014} for verification of stochastic cyber-physical
systems.  Reference \cite{Karaman2008} proposes a solution to the
vehicle routing problem with respect to MTL specifications. Reference
\cite{JunLiu2014} develops an abstraction technique and a method of
transforming MTL formulas to LTL formulas. As a result, existing
synthesis methods for discrete deterministic systems with LTL
constraints can be applied to design switching protocols for
continuous-time deterministic systems in dynamical environment subject
to MTL constraints. Reference \cite{VasuHSCC} proposes a reactive
synthesis method to non-deterministic systems with respect to
maximizing the robustness of satisfying a specification in signal
temporal logic, which is a subclass of MTL. The robustness of a path
is measured by the distance between this path and the set of paths
that satisfy the specification.
% In this paper, we study the stochastic optimal control problems with
% respect to the objective of maximizing the probability of satisfying
% a given MTL specification.
Our work differs from existing ones in both the problem formulation
and control objective.  We deal with systems with stochastic dynamics,
rather than non-deterministic systems \cite{VasuHSCC,JunLiu2014}.  We
consider reachability objectives specified in \ac{mitl}, which is a
subclass of MTL. The optimality of control design is evaluated by the
probability of satisfying the given specification. The synthesis
method is with respect to quantitative criteria (the probability of
satisfying the formula), not qualitative criteria (whether the formula
is satisfied).

% approach overview
% continuous-time continuous-state
Our solution approach utilizes the \emph{Markov chain approximation
  method} \cite{kushner2001numerical} to generate a discrete
abstraction in the form of a \ac{mdp} approximating the
continuous-time stochastic system. Based on a product operation
between the discrete abstraction and a finite-state automaton that
represents the desirable system property, a near optimal policy with
respect to the probability of satisfying the formula in the
\emph{point-based semantics} of \ac{mitl} \cite{henzinger1991temporal}
can be computed by solving an optimal planning problem in the
\ac{mdp}. We show that as the discretization gets finer in both state
space and time space, the optimal control policy in the abstract
system converges to the optimal one in the original stochastic system
with respect to the probability of satisfying the \ac{mitl} formula in
the \emph{continuous or dense-time semantics} \cite{Alur1996}.
%What is the challenge we aim to tackle?

% Why it is interesting?

% Is there any existing result relevant to our problems?

% approach?

% conclusion.

\section{Preliminaries and problem formulation}
\label{sec:prelim}
% Notations: $1_X(\cdot)$ is the indicator function defined by $1_X(x)=1$ if and
% only if $x\in X$, otherwise, $1_X(x)=0$.

% For any open space $Y$, let
% $\partial Y$ denote the boundary of $Y$ and
% $\bar Y= Y\cup \partial Y$ denotes the completion of $Y$.
\subsection{The system model and timed behaviors}
% Suppose that $w(\cdot)$ is a Wiener process on some probability space
% $(\Omega, \calF, P)$. Let $U\subseteq \bbR^k$ be a compact space.  

We study stochastic dynamical systems in continuous time.
The state of the system evolves according to the \ac{sde}
\begin{equation}
\label{sde}\text{SDE}:
\left\{ 
\begin{array}{l}
  dx(t)= f(x(t),u(t))dt + g(x(t))dw,\\
  x(0)=x_0,
\end{array}\right.
\end{equation}
where $f: X\times U \rightarrow \bbR^n$ and
$g: X \rightarrow \bbR^{n\times k}$ are continuous and bounded
functions given $X$ and $U$ as compact state and input
space; $w(\cdot)$ is an $\mathbb{R}^k$-valued, $\calF_t$-Wiener
process which serves as a ``driving noise'' and is defined on the
probability space $(\Omega, \calF, P)$; $x(\cdot)$ is an $X$-valued,
$\calF_t$-adapted, measurable process also defined on
$(\Omega, \calF, P)$ and $u(\cdot) $ is an \emph{admissible control
  law}, i.e., a $U$-valued, $\calF_t$-adapted, measurable process
defined on $(\Omega, \calF, P)$. % For
% $t \ge 0$, $u(t)$ takes value in a compact set $U$ and $x(t)$ takes
% value in a bounded set $X$.
We say $x(\cdot), u(\cdot)$ solve the SDE in \eqref{sde} provided that
\begin{equation}
  \label{integral}x(t)=x(0)+\int_0^t f(x(\tau), u(\tau))d\tau+ \int_0^t g(x(\tau))dw(\tau),
\end{equation}
holds for all time $t\ge 0$.

% Fix the sample space variable $\omega \in \Omega$, given the processes
% $x(\cdot), u(\cdot)$, we obtain sample paths $x(\cdot,\omega)$ and
% $u(\cdot,\omega)$.

% that solves
% \[
% X_t=X_0+\int_0^t b(X_\tau, u(\tau))d\tau+ \int_0^t g(X_\tau)dw(\tau)
% \]

% For each state $x\in X$, $U(x)$ is a set of admissible control
% inputs at $x$.

% that satisfies the Markov property
% \begin{multline*}
% P(X_t\in B \mid X_r= x)= \\P(X_t\in B\mid X_{r_1}=x_1,\ldots,
% X_{r_n}=x_n, X_r=x) 
% \end{multline*}
% for any Borel subset $B$ of $\mathbb{R}^n$ and time instants
% $0\le r_1\le \ldots \le r_n\le r\le t$. Let $p(r,x; t,y)$ be the
% \emph{transition density function} of the chain. It holds that
% $P(X_t \in B\mid X_r=x) = \int_B p(r,x;t,y)dy$.

% \subsection{Timed behaviors}
We introduce a labeling function that relates a sample path of the
\ac{sde} in \eqref{sde} to a \emph{timed behavior}. Let $\calAP$ be
a finite set of atomic propositions and $L: X \rightarrow 2^{\calAP}$
be a labeling function that maps each state $x\in X$ to a set of
atomic propositions that evaluate true at that state.

A \emph{time interval} $I$ is a convex set $\langle t, t'\rangle$
where $t,t'\in\bbR_{\ge 0}$, the symbol `$\langle$' can be one of
`$($', `$ [$', and the symbol `$\rangle $' can be one of `$)$',
`$ ]$', and $t\le t'$. For a time interval of the above form, $t$ and
$t'$ are left and right end-points, respectively.  A time interval is
empty if it contains no point. A time interval is \emph{singular} if
$t=t'$ and it contains exactly one point.
\begin{definition} \cite{Furia2006}
  A \emph{dense-time behavior} over an infinite-time domain
  $[0, \infty)$ and a set $\calAP$ of atomic propositions is a
  function $b: [0, \infty)\rightarrow 2^{\calAP}$ which maps every
  time instant $t \ge 0$ to a set $b(t)\in 2^{\calAP}$ of atomic
  propositions that hold at $t$.
\end{definition}
Given a continuous sample path $x(\cdot,\omega), \omega \in \Omega$ of
the stochastic process $x(\cdot)$, the timed behavior $b$ of this
sample path is $b(t)= L(x(t,w))$, for all $t\ge 0$.
% Given a stochastic process $\{x(t), t\ge 0\}$ and a labeling function
% $L:X\rightarrow 2^{\calAP}$, we obtain a stochastic process
% $\{L(x(t)), t\ge 0 \}$ where $L(x(t))$ represents the set of atomic
% propositions that evaluate true by chance at time $t$.
% where $L(x(t))$ is a random
% variable. %\emph{timed behavior} is $b(t)= L(x(t))$.
% tocheck: remove the interval.
%   For a given time interval $I$, we denote $b(I)$ the
% timed behavior for interval $I$.
% remove the definition of behavior process.
% Given a stochastic process $\{X_t, t \ge 0\}$ induced
% from the stochastic system with a control $u(t), t\ge 0$, the
% stochastic process of timed behavior is $\{B_t, t\ge 0\}$ where
% $B_t= L(X_t)$.

\begin{definition}
  \cite{Furia2006} Let $b$ be a dense-time behavior and
  $\delta \in \mathbb{R}_{>0}$ be a positive real number, referred to
  as the \emph{sampling interval}. The \emph{canonical sampling}
  $b^\delta =\{b^\delta_n, n \in \bbZ \}$ of the timed behavior $b$ is
  defined such that $b^{\delta}_n= b( n\delta).$
\end{definition}

% \textcolor{red}{Does the interval has to be with endpoints in $\nat$? }

% \paragraph*{\bf Dense-time and discrete-time semantics} We consider an
% \ac{mitl} formula $\varphi$ that can be interpreted over both dense-
% and discrete-time behaviors, denoted as $[\varphi]_\bbR$ and
% $[\varphi]_\bbZ$  respectively.

\subsection{Specifications}

We introduce metric interval temporal logic \cite{Alur1996}, a
subclass of MTL, to express system specifications.
\begin{definition}[Metric interval temporal logic]
Given a set $\calAP$ of atomic propositions, the formulas of \ac{mitl}
are built from $\calAP$ by Boolean connectives and time-constrained
versions of the \emph{until} operator $\cal U$ as follows.
\[
\varphi\coloneqq \top \mid \perp \mid \varphi_1\land \varphi_2 \mid \neg
\varphi\mid p \mid \varphi_1 \mathcal{U}_I \varphi_2
\]
where $p \in \calAP$, $I $ is a \emph{nonsingular} time interval with
integer end-points, and $\top$, $\perp$ are unconditional true and
false, respectively.
\end{definition}

% There are two possible semantics for \ac{mitl}, one of which is called
% \emph{dense-time} or \emph{continuous} and the other one of which is
% called \emph{point-based}. 

\paragraph*{Dense-time semantics of \ac{mitl}} Given a timed behavior
$b$, we define $b,t \models \varphi$ with respect to an \ac{mitl}
formula $\varphi$ at time $t$ inductively as follows:
\begin{itemize}
\item $b,t\models p$ where $p\in \calAP$ if and only if $ p \in b(t)$;
\item $b,t \models \neg \varphi $ where $b,t \not \models \varphi$;
\item $b,t \models \varphi_1\land \varphi_2$ if and only if $b,t\models
  \varphi_1$ and $b,t\models \varphi_2$;
\item $b,t\models \varphi_1 \calU_I \varphi_2$ if and only if there
  exists $t'\in I$ such that $b,t+t' \models\varphi_2$ and for all
  $t''\in [0,t')$, $b, t+t'' \models \varphi_1$;
\end{itemize}
We write $b\models \varphi$ if $b, 0\models \varphi$. We also define
temporal operator $\lozenge_I \varphi= \truev \calU_I \varphi$
(eventually, $\varphi$ will hold within interval $I$ from now) and
$\square_I \varphi = \neg (\lozenge_I \neg \varphi)$ (for all
points within $I$, $\varphi$ holds.)
 % The two semantics shares rules for basic
% modalities and only differ in the interpretation of \emph{positions}:

\subsection{Timed automata}
An \ac{mitl} formula can be translated into equivalent
non-deterministic timed automaton \cite{Alur1996}. We consider a
fragment of \ac{mitl} which can be translated into equivalent
\emph{deterministic} timed automaton.
% We consider a subclass of MTL formulas which do not have punctual time
% intervals and unbounded intervals, and can be translated into
% equivalent deterministic \emph{timed automata} \cite{Alur1996}.
% do not consider MTL formulas that have punctual time intervals and
% unbounded intervals in the until operator. This subclass of MTL is
% called \ac{mitl} \cite{Alur1996}.
% We consider a subclass of
% \ac{mitl} formulas which can be represented as languages (a subset of
% \emph{timed words}) accepted by deterministic \emph{timed automata}.

Let $\Sigma$ be a finite alphabet. $\Sigma^\ast, \Sigma^\omega$ are
the sets of finite and infinite words (sequences of symbols) over
$\Sigma$. A \emph{(infinite) timed word} \cite{Alur1994183} over
$\Sigma$ is a pair $w=(\tau,\sigma)$, where
$\sigma=\sigma_0\sigma_1\ldots \in \Sigma^\omega$ is an infinite word
and $\tau= \tau_0\tau_1\ldots $ is an infinite \emph{timed sequence},
which satisfies \begin{inparaenum}[1)]
\item \emph{Initialization}: $\tau_0= 0$;
\item \emph{Monotonicity}: $\tau$ increases strictly monotonically;
  i.e., $\tau_i<\tau_{i+1}$, for all $i\ge 0$;
\item \emph{Progress}: For every $n \ge 1$ and
  $ \tau_0 \le t< \tau_n$, there exists some $i\ge 0$, such that
  $\tau_i>t$.
\end{inparaenum}
The conditions ensure that there are finitely many symbols (events) in
a bounded time interval, known as \emph{non-Zenoness}.  We also write
$w=(\tau,\sigma)= (\tau_0,\sigma_0)(\tau_1,\sigma_1)\ldots$.

% \textcolor{red}{Emphasize the non-zeno property.}

% Note that there is
% no priori bound on the length of a finite timed word. We also write
% $(\sigma,\tau) = (\sigma_0,\tau_0)\ldots (\sigma_1,\tau_1)$.

% \begin{definition}[Timed state sequence]
%   A \emph{finite timed state sequence} in a stochastic system with
%   labeling function $L$ is a tuple $\langle x, \tau \rangle$ where
%   for some $n\in \nat$, $ x = x_1x_2\ldots x_n$ is a sequence of
%   states, $ \tau=\tau_1\tau_2\ldots \tau_n$ is a sequence of time
%   stamps.
% \end{definition}
% Each finite timed state sequence $\langle x, \tau\rangle$ under the
% labeling function is related with a finite timed word
% $\langle\sigma, \tau \rangle$ where
% $\sigma=\sigma_0\sigma_1\ldots \sigma_n$ is a finite sequence of sets of
% atomic propositions such that $\sigma_i =L(x_i)$, for all
% $1\le i \le n$.

Before the introduction of timed automata, we introduce \emph{clock}
and \emph{clock} constraints: Let $C$ be a finite set of clocks,
$C= \{c_1,c_2,\ldots, c_M\}$. We define a set $\Phi_C$ of clock
constraints over $C$ in the following manner. Let $k \in \nat$ be a
non-negative integer, and $\bowtie \in \{=,\ne , <, >, \ge, \le \}$ be
a comparison operator,
\[
\varphi\coloneqq \top \mid \perp \mid c \bowtie k \mid c-c' \bowtie k
\mid \varphi_1\land\varphi_2\mid \varphi_1\lor \varphi_2,
\]
where $c, c'\in C$ are clocks.
\begin{definition}\cite{Alur1994183}
  A \emph{deterministic timed automaton} is a tuple
  $\calA=\langle Q, 2^{\calAP}, \init, F, C, T \rangle$ where $Q$ is a
  finite set of states, $2^{\calAP}$ is a finite set of alphabet with
  the set $\calAP$ of atomic propositions, $\init$ is the initial
  state, $F$ is a finite set of accepting states, $C$ is a finite set
  of clocks. The transition function
  $T: Q\times 2^\calAP\times \Phi_C \rightarrow Q \times 2^C $ is
  deterministic and interpreted as follows: If
  $T(q, a, \phi)=(q', C')$ then $\calA$ allows a transition from
  $q$ to $q'$ when the set $a \in 2^{\calAP}$ of atomic propositions
  evaluate true and the clock constraint $\phi \in \Phi_C$ is
  met. After taking this transition, the clocks in $C'\subseteq C$ are
  reset to zero, while other clocks remain unchanged.
\end{definition} 

% A timed automaton is \emph{deterministic} whenever for every timed
% word $(\sigma, \tau)$, there is at most one run from the initial state
% $(q_0, 0)$ which reads $(\sigma, \tau)$.
% A \emph{clock valuation} is a function
% $\mu : C\rightarrow \mathbb{R}_{\ge 0}$ that outputs the clock value
% $\mu(c)$ for each clock $c\in C$.

% \begin{assumption}

% \end{assumption}

For each clock $c_i \in C$, we denote $\calV_i$ the range of that
clock.  For notational convenience, we define a \emph{clock vector}
$v \in \bbR^M$ where the $i$-th entry $v[i]$ of the clock vector $v$
is the value of clock $c_i$, for $i\in \{1,2,\ldots, M\}$. Given
$t\in \mathbb{R}_{\ge 0}$, let
$v\oplus t= (v[1]+t, v[2]+t,\ldots, v[M]+t)$.  We use $\bm 0$ for the
clock vector $v$ where $v[i]=0$ for all $i \in \{1,2,\ldots, M\}$ and
$\calV = \Pi_{i=1,\ldots,M} \calV_i$ the set of all possible clock
vectors in $\calA$. Note that a clock vector is essentially a
\emph{clock valuation} defined in \cite{Alur1994183}.

A \emph{configuration} of $\calA$ is a pair $(q, v)$ where $q$ is a
state and $v$ is a clock vector. A transition
$T(q,a, \phi)= (q', C')$ being taken from the configuration $(q,v)$
after $\delta$ time units is also written as
$(q,v)\xrightarrow{\delta, a}(q',v')$ where
$v \oplus \delta \models \phi$, and $v'[i]= v[i]+\delta$ if
$c_i\notin C'$, otherwise $v'[i]=0$.

A \emph{run} in $\calA$ on a timed word
$w=(\tau_0, a_0 )(\tau_1,a_1)\ldots $ is an infinite alternating
sequence of configurations and delayed transitions
$\rho=(\init, \bm 0)\xrightarrow{\Delta \tau_0, a_0}
(q_0,v_0)\xrightarrow{\Delta \tau_1, a_1}(q_1,v_1) \ldots $,
with $\Delta \tau_0=\tau_0$ and $\Delta\tau_i = \tau_{i}-\tau_{i-1}$
for $i \ge 1$, subject to the following conditions:
\begin{enumerate}

\item There exists $C_0 \subseteq C$ and $\phi_0 \in \Phi_C$ such
  that $\bm 0 \oplus \tau_0 \models \phi_0$,
  $T(\init, a_0, \phi_0 ) = (q_0, C_0)$ and $v_0[i]= \tau_0$ for
  all $c_i \notin C_0$ and $v_0[i]=0$ for all $c_i\in C_0$.
\item For each $ i \ge 0 $, there exist $C_{i+1}\subseteq C$ and
  $\phi_{i+1}\in \Phi_C$ such that $v_i \oplus \Delta \tau_{i+1}$
  satisfies the clock constraint $\phi_{i+1}$,
  $T(q_i, a_{i+1}, \phi_{i+1})=(q_{i+1}, C_{i+1})$ is defined and
  $v_{i+1}[k]=v_{i}[k]+\Delta\tau_{i+1}$ for all $c_k\notin C_{i+1}$
  and $v_{i+1}[k]=0$ for all $c_k\in C_{i+1}$.
\end{enumerate} 
% Given a finite timed word $w$, if it triggers a transition in the
% timed automaton from $(q,v)$ to $(q',v')$, we also write
% $(q,v)\stackrel{w}{\hookrightarrow} (q',v')$.

We consider reachability objectives: A run $\rho$ on a timed word $w$
is \emph{accepting} if and only if $\Occ(\rho)\cap F\ne \emptyset$
where $\Occ(\rho)$ is the set of states in $Q$ occurring in
$\rho$. The set of timed words on which runs are accepted by $\calA$
is called the language of $\calA$, denoted $\mathcal{L}(\calA)$.

\begin{example}
  % Consider a simple robot motion planning example in which the robot
  % is required to visit region $R_1$ within the time interval
  % $I= [3,5]$.
As a simple example of timed automata, let $\calAP=\{R_1\}$ and 
  % Let $\calAP= \{R_1\}$ % where $R_1$ means the robot is in
  % region $R_1$.
the specification formula is $ \varphi = \lozenge_{[3,5]} R_1$. The
reachability specification can be expressed with a deterministic timed
automaton $\calA_\varphi$ in Figure~\ref{fig:reach}. The set of final
states is $F=\{q_1\}$.  The timed automaton accepts a timed word with
a prefix $(0, \{\neg R_1\})(3.5, \{R_1\})$, i.e.,
$w = (0, \{\neg R_1\}) (3.5, \{R_1\}) \ldots $ since
$(\init,0)\xrightarrow{0, \{\neg R_1\}}(\init,0) \xrightarrow{3.5,
  \{R_1\}} (q_1,0)$
and $q_1$ is accepting. It does not accept
$w= (0, \{\neg R_1\}) (2.8,\{R_1\}) \ldots$,
$w=(\tau, (\{\neg R_1\})^\omega )$ for an arbitrary timed sequence
$\tau$, or $w=(0, \{\neg R_1\})(6,\{R_1\})\ldots$ because either $R_1$
is evaluated true when $c<3$ or $c > 5$, or it is never true over an
infinite timed sequence.
% because either
  % the robot reached $R_1$ but not within the interval, or it never
  % reaches $R_1$.
\begin{figure}[h]
\centering
\includegraphics[width=0.4\textwidth]{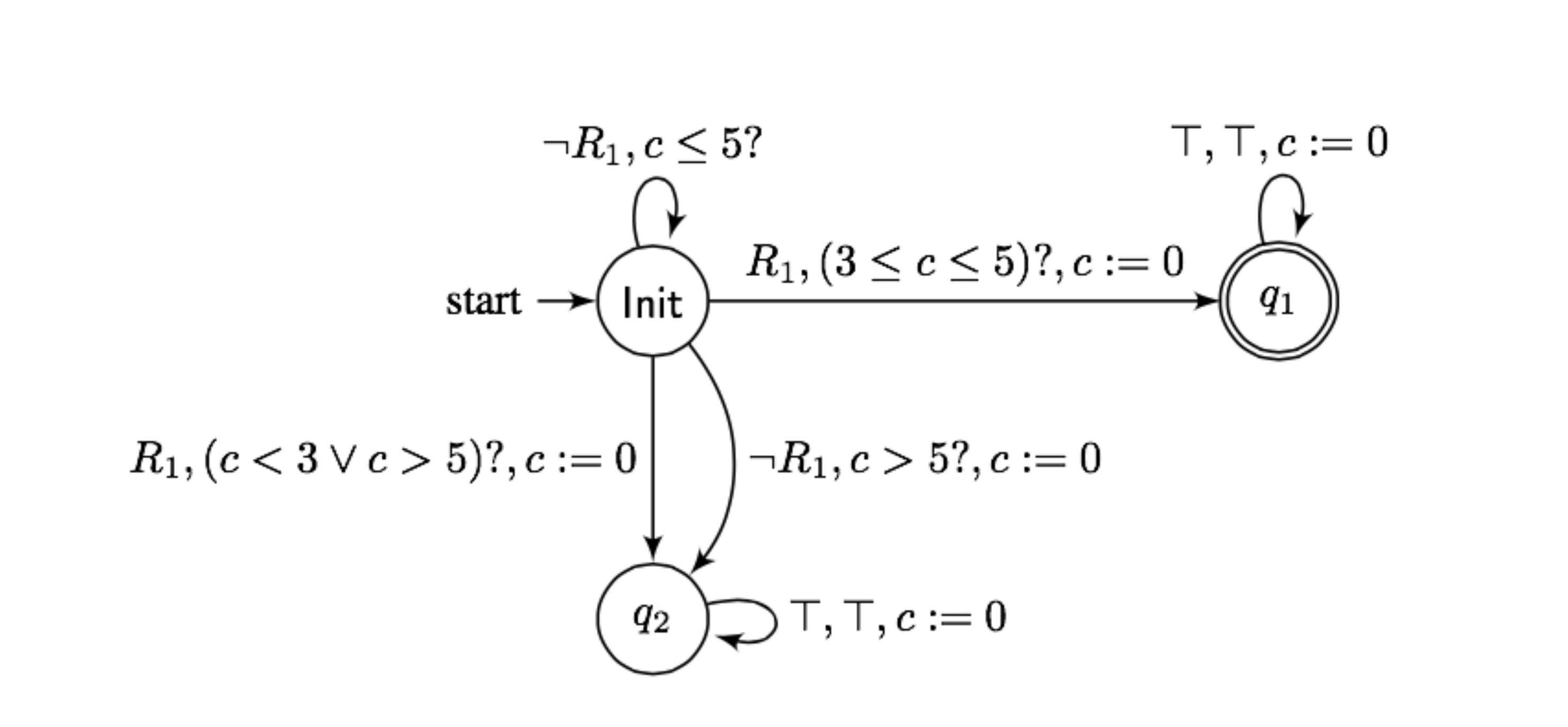}
\caption{Timed automaton $\calA_\varphi$ for
  $\varphi = \lozenge_{[3,5]} R_1$. Only one clock $c$ is used. A
  transition labeled $(a, \phi_c)$ is taken if and only if both
  $\phi_c$ and $\{a \}$ evaluate true. A transition labeled
  $(a, \phi_c,c:=0)$ is taken if and only if both $\phi_c$ and
  $\{a\} $ evaluate true and along with taking the transition, the
  clock $c$ is reset to $0$.}
\label{fig:reach}
\end{figure}
\end{example}
% Given an \ac{mitl} formula $\varphi$, we can construct a timed
% automaton $\calA_\varphi$ precisely accepts the timed words that
% satisfies $\varphi$ \cite{Bouyer2009}. 

%Since we only consider 

% \begin{example}
% \end{example}

% It is noted that in this paper, we consider the trace sampled from the
% continuous state trajectory of the \ac{sde}. In numerical methods, we
% approximate the stochastic differential equations by discrete-state
% Markov decision processes which depends on the sizes of the time step
% and spatial steps. Then there exists a bounded error between the
% finite-time continuous trajectory and the sampled discrete trajectory
% for a given control input. As the step size goes arbitrarily to zero,
% the error goes arbitrarily close to zero.

\subsection{Problem formulation}
Given a sampling interval $\delta$ and a timed behavior
$b:[0,\infty)\rightarrow 2^{\calAP}$, we map the canonical sampling
$b^{\delta}$ of $b$ to a timed word
$ \calT(b^\delta) = w= (0,\sigma_0)(\delta, \sigma_1) \ldots$ such
that for any $i\ge 0$, $\sigma_i = b(i\delta)$. We say that the timed
behavior $b$ satisfies the formula $\varphi$ in the \emph{point-based
  semantics} under the sampling interval $\delta$, denoted
$b^\delta \models \varphi$, if and only if $\calT(b^\delta)$ is
accepted in the timed automaton $\calA_\varphi$ that expresses
$\varphi$. % The timed behavior
% $b\models \varphi$ in the \emph{continuous or dense-time semantics}
% if and only if $b^\delta the sampling interval
% $\delta$ is a positive infinitesimal.
The sampling interval $\delta$
determines a sequence of positions (time instances) $0,
\delta, 2\delta,\ldots$ in the timed behavior. With $\delta
$ being a positive infinitesimal, any position in a timed behavior
$b$
appears in the timed sequence of the timed word $\calT(
b^\delta)$. Thus, we say that the timed behavior
$b$
satisfies $\varphi$
in the \emph{continuous or dense-time semantics}, i.e., $b\models
\varphi$, if and only if $ \lim_{\delta\rightarrow 0}
\calT(b^\delta)\in
L(\calA_\varphi)$. A formal definition of satisifiability of MTL
formulas over dense-time and point-based semantics is given in
\cite{Bouyer2009} and the relation between these two semantics has
been studied in \cite{Furia2006}.
We say that a sample path of the \ac{sde} in \eqref{sde} satisfies an
\ac{mitl} formula $\varphi$ in the dense-time semantics
(resp. point-based semantics under the sampling interval $\delta$) if
its timed behavior satisfies $\varphi$ in the dense-time semantics
(resp. point-based semantics under the sampling interval
$\delta$). Formally, let $x(\cdot, w)$ where $w\in \Omega$ be a sample
path of the stochastic process $\{ x(t), t \ge 0 \}$. We have that
$\lim_{\delta\rightarrow 0} [L(x(\cdot,w))]^\delta \models \varphi$ is
equivalent to $L(x(\cdot,w)) \models \varphi$.

% A continuous state trajectory that satisfies $\varphi$ in
% the point-based semantics with the sampling interval $\delta $
% satisfies $\varphi$ in the dense-time semantics as
% $\delta\rightarrow 0$.

% We say a finite (or infinite) state trajectory $x(t), t\ge 0 $
% \emph{satisfies an \ac{mitl} formula $\varphi$ in the dense-time
%   semantics} if and only if $\calT(L(x(t)))$ is accepted in the
% timed automaton $\calA_\varphi$. We denote
% $x(t)\models [\varphi]_\bbR$ for the satisfaction in the dense-time
% semantics.

% On the other hand, $x(t), t\ge 0 $
% \emph{satisfies an \ac{mitl} formula $\varphi$ in the point-based
%   semantics} if and only if $\calT(L^{\delta}(x(t)))$ is accepted
% in the timed automaton $\calA_\varphi$. We denote
% $x(t)\models [\varphi]_\bbZ$ for the satisfaction in the point-based semantics.

Given a stochastic process $x(\cdot)$ and an admissible control law
$u(\cdot)$ that solve the \ac{sde} in \eqref{sde}, the probability of
satisfying a formula $\varphi$ in the system under the control law
$u(\cdot)$ is the sum of probabilities of continuous sample paths of
$x(\cdot)$ that satisfy the formula $\varphi$ in the dense-time or
point-based semantics (with respect to a given sampling interval).

\begin{problem}
\label{problem}
Given an \ac{sde} in \eqref{sde} and a timed automaton
$\calA_\varphi=\langle Q, 2^{\calAP},\init, F, C, T \rangle$
expressing an \ac{mitl} formula $\varphi$, compute a control input
$u(\cdot)$ that maximizes the probability of satisfying $\varphi$ in
the dense-time semantics.
\end{problem}
% They showed that the satisfiability and model checking problems for
% MITL relative to a dense-time semantics are also EXPSPACE-complete.
% \begin{problem}
% \label{problem2}
% Given a stochastic system in \eqref{sde} and a timed automaton
% $\calA_\varphi=\langle Q, 2^{\calAP},\init, F, C, T \rangle$ expressing
% an \ac{mitl} formula $\varphi$,
% compute a control input $u(\cdot)$ that maximizes the probability of
% satisfying $\varphi$ in the point-based semantics with a given
% sampling interval $\delta$.\end{problem}

% \begin{example}

% \end{example}
\section{Main result}
In this section, we first show that for the \ac{sde} in \eqref{sde},
Problem~\ref{problem} can be formulated as a stochastic optimal
control problem in a system derived from the \ac{sde} with an
augmented state space for capturing relevant properties with respect
to its \ac{mitl} specification. Then, we introduce a numerical scheme
that computes an optimal policy in a discrete-approximation of the
\ac{sde} in \eqref{sde} with respect to the probability of satisfying
the specification in the point-based semantics. The numerical scheme
is based on the so-called Markov chain approximation method
\cite{kushner2001numerical}.  We prove that such a policy converges to
a solution to Problem~\ref{problem} as the discretization gets finer.

We make two assumptions.
\begin{assumption}
\label{assumption0}
  The state space $X$ and the clock vector space $\calV$ are bounded.
\end{assumption}
This condition ensures a finite number of states in the discrete
approximation. In certain cases, we might also require $U$ to be
bounded in order to approximate the input space with a finite set.

\begin{assumption}\label{assumption1}
  $f(\cdot)$ and $g(\cdot)$ are bounded, continuous, and Lipschitz
  continuous in state $x$, while $f(\cdot)$ is uniformly so in $u$.
\end{assumption}
Assumption~\ref{assumption1} ensures that the \ac{sde} in \eqref{sde}
has a unique solution for a given controller $u(\cdot)$.
% \begin{assumption}
%   In the timed automaton $\calA_\varphi$, $\calV$ is bounded.
%  \end{assumption}
% \begin{assumption}
% For all $x\in X$, $U(x)=U$.
% \end{assumption}

% \begin{definition}
%   Let
%   $\Pi=\{u: X\rightarrow U\mid u \text{ is Borel measurable and } u(x)
%   \in U(x), \forall x \in X\}$.
%   Let $\Pi^\infty$ be the set of all infinite sequences
%   $\pi=(u_0,u_1,\ldots)$ of elements in $\Pi$. Each
%   $\pi\in \Pi^{\infty} $ is called a \emph{Markov policy}, which
%   chooses a control input as a function of the discrete time step and
%   of the current state. That is, at the $n$-th time step the control
%   input $u_n(x)$ is applied if the current state is $x$.  A
%   \emph{stationary policy} $\pi=\{u,u,\ldots\}$ is a Markov policy with $u=u_i=u_j$,
%   for all $i,j \ge 0$.
% \end{definition}
%  Let $\delta$
%   be the time step that discretize the continous time. The continuous
%   interpolation $u(t)$ of a discrete-time policy $\pi$ with
%   $\delta$ being the time step is $u(t)= u_n(x)$, for
%   $t\in [n\delta, (n+1)\delta]$.

\subsection{Characterizing the reachability probability}
For reachability objectives in \ac{mitl}, Problem~\ref{problem} is
also referred to as the \emph{probabilistic reachability problem}.

%%%%%%%%%remove the notion of behavior equivalence %%%%%
%  A timed word
  % $w =(a_0,\tau_0)(a_1,\tau_1)\ldots $ over alphabet $2^\calAP$ can be
  % mapped into a dense-time behavior
  % $b:[\tau_0,\infty)\rightarrow 2^\calAP$ with an offset $\tau_0$ such
  % that for $i \ge 0$, for any $t \in [\tau_i, \tau_{i+1})$,
  % $b(t)=a_i$. We denote $b^w$ the dense-time behavior for the timed
  % word $w$.

% \begin{definition}
%   Two timed words $w_1,w_2$ are \emph{behavior-equivalent} if and only
%   if $b^{w_1}=b^{w_2}$.  The set of timed words behavior-equivalent to
%   $w$ is denoted $[w]_b$. \end{definition}

% Given a behavior $b$ and a configuration $(q,v)$, if a run in $\calA$
% on any timed word $w$ of behavior $b$ triggers a transition to the
% configuration $(q',v')$ from $(q,v)$, then we denote
% $(q,v) \hookrightarrow_{b} (q',v')$.

% \begin{lemma}
%   If two words are behavior-equivalent, then $w_1\in L(\calA)$ if
%   and only if $w_2\in L(\calA)$.
% \end{lemma}

% Correspondingly, the timed word obtained from a dense-time behavior
% for the interval $[t_1,t_2]$ where $t_1<t_2$ is
% $\calT(b_{t_1:t_2})$ is a finite timed word.

% \begin{definition}
%   Given a fixed control $u(t), t\ge 0$, the controlled system
%   described by \eqref{sde} is a stochastic process induced by $u$,
%   denoted as $\{X_t, t\ge 0\}$.  
% \end{definition}

% Let $x(t)$ be a sampled path in the
  % stochastic process and
A state in $X\times Q\times \calV$ is called a \emph{product state},
following from the fact that it is a state in a product construction
between the stochastic process for the controlled stochastic system
and the timed automaton expressing the specification. We define a
projection $\pi_i$ such that for a given tuple $s$, $\pi_i(s)$ is the
$i$-th element in the tuple. The projection $\pi_i$ is extended to
sequences of tuples in the usual way:
$\pi_i(s\rho)= \pi_i(s)\pi_i(\rho)$ where $s$ is a tuple and $\rho$ is
a sequence of tuples.

Let $S= X\times Q\times \calV$. For a stochastic process
$\{x(t), t\ge 0 \}$, we derive a \emph{product stochastic process}
$\{s(t), t\ge 0\}$ where $s(t)= (x(t), q(t), v(t))$ is a random
variable describing the product state. The process $\{s(t), t\ge 0
\}$ % vector of
% random variables describing a state in the stochastic system, a state
% in the timed automaton and a clock vector at time $t$, corresponding
% to the process $\{x(t), t\ge 0\}$ induced by an admissible control
% process $\{u(t), t\ge 0\}$ in the stochastic system \ref{sde}
satisfies the following conditions.
\begin{itemize}
\item $s(0)= (x(0), q(0),v(0))$ where $v(0) =\bm 0$ and
  $(\init,\bm 0 )\xrightarrow{0, L(x(0))}(q(0),\bm 0)$.
\item For any time $\tau \in [0,\infty)$, let
  $\delta= \inf_t\big(\exists \phi, v(\tau)\oplus t \models \phi
  \text{ and } T \text{ is defined for }(q(\tau), L(x(\tau+t)),
  \phi)\big)$.
  If $T(q(\tau), L(x(\tau+\delta)), \phi)=(q',C')$, then let
  $q(\tau+\delta)=q'$, $v(\tau+\delta)[i]=v(\tau)[i]+\delta$ for
  $c_i \notin C'$ and $v(\tau+\delta)[i]=0$ for $c_i\in C'$. Moreover,
  for all $\tau \le t < \tau + \delta$,
  $q(t)=q(\tau), v(t)=v(\tau)\oplus t$.
\end{itemize}

Alternatively, given a sample path $x(\cdot,\omega)$,
$\omega \in \Omega$, suppose that at time $\tau$ the configuration in
$\calA_\varphi$ is $(q,v)$, the labeling $L(x(\tau+\delta,\omega))$
and the clock vector $v\oplus \delta$ trigger a transition precisely
at time $\tau+\delta$ and between the interval $[\tau, \tau+\delta)$,
no transition is triggered. Then, the configuration in $\calA_\varphi$
changes from $(q,v)$ to $(q',v')$ also at time $\tau+\delta$ provided
that $(q,v) \xrightarrow{\delta, L(x(\tau+\delta, \omega))} (q',v')$.
Moreover, for any time $t$ during the time interval
$ \tau \le t < \tau+\delta$, the state in the specification automaton
remains to be $q$ and each clock increases by $t$ as the time passes.

For a measurable function $f$ that maps sample paths in the process
$s(\cdot)$ into reals, we write $E_s^u(f)$ for the expected value of
$f$ when the initial state is $s(0)=s$.

The following lemma is an immediate consequence of the derivation
procedure for the product stochastic process.
\begin{lemma}
\label{lma:relate}
  Given a set $G= X\times F\times \calV$, let $P_{x}(\varphi )$ denote
  the probability of a sample path in the stochastic process
  $\{x(t),t\ge 0\}$ starting from $ x (0)=x$ and satisfying $\varphi$
  in the dense-time semantics and $P_{s}( G)$ is the probability of
  reaching the set $G$ in the derived product stochastic process
  $\{s(t), t\ge0 \}$ with $s(0)=s$. It holds that
  $P_{x}(\varphi )= P_{s}( G). $
\end{lemma}
By Lemma~\ref{lma:relate}, we can define a value function in the product
stochastic process to characterize the probability of satisfying
$\varphi$ in the dense-time semantics.
\paragraph*{\bf Dense-time reachability probability}
% Given a target set $G $ and processes $\{x(t), t\ge0\}$,
% $\{ u(t), t\ge 0\}$ that solves \eqref{integral}, let the derived
% product stochastic process be $\{s(t), t\ge 0\}$.
The probability of
reaching $G$ from a product state $s \notin G$ under a controller
$u(\cdot)$ is denoted as $P_{s, u}( G)$.
% \begin{multline*}
% P_{s,u}^{\infty}( G)  = P_{s,u}\{s^u(t)\in G \text{ for some }
% t \ge 0, \\\text{ and } \forall t' \le t, s^u(t') \notin G\}.
% \end{multline*}
% We assume that the boundary of $X$ are sufficiently smooth. In the
% product state space we add a new state $\sink$ and define the
% transition measure
% $T: S \cup \{\sink\} \times \calB(S \cup \{\sink\} ) \rightarrow \bbR$
% as
% \[
% \hat p(t,s;r, \sink)=\left\{ \begin{array}{ll}
% 1, & \text{ if } s\in G\\
% 0, & \text{ otherwise}.
% \end{array}
% \right.
% \]
% With slightly abusing of notation, we denote the modified transition
% density function also by $\hat p(\cdot)$.  Also, 
We construct a reward function $r: S \rightarrow \{1,0\}$ such
that $r(s)=1_G(s)$ where $1_A(\cdot)$ is the indicator function, i.e.,
$1_A(x)=1$ if $x\in A$, and $1_A(x)=0$ otherwise.
% \begin{equation}
% \label{eq:reward}
% r(s)=\left\{ \begin{array}{ll}
% 1, & \text{ for all } s\in G,\\
% 0, & \text{ otherwise}.
% \end{array}\right.
% \end{equation}
Then, $P_{s,u}( G)$ is evaluated by the value function
\[
P_{s,u}( G) = W(s, u)= E^u_{s}\left\{ \int_0^T r(s(t) )dt \right\},
\]
where $T$ is a random variable describing the stopping time such that
$T= \inf_{t\ge 0}(s(t) \in G)$.  % We may also write
% \begin{multline}
% \label{eq:backwardrec}
% W(s, u)= \\
% E_{s}^u \left\{\int_{0}^{\tau \land T}r(s(t)))dt \cdot 1_{\tau \le T}
%   + W(s(\tau), u) 1_{\tau >T }\right\},
% \end{multline}
% where $1_p$ equals $1$ if the propositional formula $p$ evaluates
% true, otherwise $1_p=0$.

%  and $E_{s_0}(\cdot)$ is the expectation
% with respect to the probability measure for the product stochastic
% process $\{s(t), t\ge0\}$.  It can be computed as the fixpoint
% $W(s,u)$ defined by the value function
% \[
% W(s,u) = 1_{G}(s) + 1_{G^c}(s) \cdot E_{s} \left[
%   \int_{0}^{T_u} 1_G(s(t)) dt
% \right]
% \]
% where $T_u= \inf\{t\ge 0 : s(t) \in G\}$ is the stopping time.  By
% definition, $W(s, u)$ is the probability of satisfying the \ac{mitl}
% formula in the dense time semantics.

% \begin{lemma}\cite{kushner2001numerical}
%   The value function $W$ can be computed, for $\Delta >0 $,
% \begin{multline}
% \label{eq:backwardrecur}
% W(s, u)= 1_G(s)+ 1_{G^c}(s) \cdot E_{s} \big[ \int_{0}^{\Delta \wedge
%     T_u} 1_G(s(t)) dt I_{\tau \le \Delta} \\+ W(s(\Delta),u)
%   I_{\tau>\Delta} \big].
% \end{multline}
% \end{lemma}
% The optimal value, where the supremum is taken over all admissible
% control policies, is
% \[
% V^\ast(s) = \sup_{u \in \Pi} W(s,u).
% \]

The optimal value function is defined as
$ V (s) = \sup_{u \in \Pi} W(s,u)$, where $\Pi$ is the set of all
admissible control policies for the \ac{sde} in \eqref{sde}.

So far, we have shown that given $x(\cdot), u(\cdot)$ that solve the
\ac{sde} in \eqref{sde}, the probability that a sample path in the
stochastic process $x(\cdot)$ satisfies the \ac{mitl} formula
$\varphi$ in the dense-time semantics can be represented by the value
$W(s(0),u)$ in the derived product stochastic process $s(\cdot)$ under
the reward function $r:S\rightarrow
\{1,0\}$.% We define a reward function such that the
% policy that solves Problem~\ref{problem} is the one that attains the
% optimal value $V(s)$.

\subsection{Markov chain approximation}
In this section, we employ the methods in \cite{kushner2001numerical}
to compute locally consistent Markov chains that approximate the SDE
in \eqref{sde} under a given control policy.

Given an approximating parameter $h$, referred to as the \emph{spatial
  step}, we obtain a discretization of the bounded state space,
denoted by $X^h$, which is a finite set of discrete points
approximating $X$. Intuitively, the spatial step $h$ characterizes the
distance between neighboring and introduces a partition of $X$. The
set of points in the same set of the partition is called an
\emph{equivalent class}.  For each $x\in X^h$, the set of points in
the same equivalent class of $x$ is denoted
$[x]=\{x'\in X\mid x \le x'< x+h\} $.  We call $x\in X^h$ the
representative point of $[x]$. % The parameter $h$ is chosen such that
% $X \subseteq \cup_{x\in X^h }[x] $.

We define an \ac{mdp} $M^h = \langle X^h , U, P^h, x_0 \rangle$ where
$X^h$ is the discrete state space. $U$ is the input space, which can
be infinite.  $P^h: X^h \times U \times X^h \rightarrow [0,1]$ is the
transition probability function (defined later in this section).  The
initial state is $x_0\in X^h$ such that the initial state $x(0)$ of
the \ac{sde} in \eqref{sde} satisfies $x(0)\in [x_0]$.

% Let $\delta$ be the discretization parameter for the spatial space
% and
\begin{definition} \cite{kushner2001numerical} Let $\Delta t_i^h$ be
  the interpolation interval at step $i$ for $i \ge 0$. Let $t_0^h=0$
  and $t_n^h= \sum_{i=0}^{n-1} \Delta t^h_i$ for $n\ge 1$ be
  interpolation times.  The \emph{continuous interpolations}
  $x^h(\cdot), u^h(\cdot)$ of the stochastic processes
  $\{x^h_n, n \in \bbZ \}$ and $\{u_n^h, n \in \bbZ\}$ under the
  interpolation times $\{ t_n^h, n \in \bbZ \}$ are
  $x^h(t) = x_n^h, u^h(t) = u_n^h, \text{ for all } t \in [t_n^h,
  t_{n+1}^{h}).$
\end{definition}

Given a policy $\{u_n, n \in \bbZ\}$, let $\{x_n, n\in \bbZ\}$ be the
induced Markov chain from $M^h$ by such a policy.  It is shown that if
a certain condition is satisfied by the spatial step and the
interpolation times, the continuous interpolations of
$\{x_n, n\in \bbZ\}$ and $\{u_n, n \in \bbZ\}$ converges to processes
$x(\cdot)$ and $u(\cdot)$ which solve the \ac{sde} in \eqref{sde}.

% todo

\begin{theorem}
\label{thm:stateconverge}
\cite{kushner2001numerical} Suppose Assumption \ref{assumption1}
holds.  % suppose \ac{mdp} $M^h$ satisfies the
% following: and the interpolation times $\{ t_n^h, n \in \bbZ \}$
For any policy $\{u_n^h, n \in \bbZ\}$, let the chain induced from
$M^h$ by this policy be $\{x^h_n, n \in \bbZ\}$.  Let $E_{x,n}^{h,a}$
denote the conditional expectation given
$\{x^h_i, u^h_i, 0 \le i<n, x^h_n=x,u^h_n=a \}$. Then, for all
$x\in X$ and $a\in U$, the chain $\{x^h_n, n \in \bbZ\}$ satisfies the
\emph{local consistency condition}:
\begin{align*}
& E_{x,n}^{h,a}(\Delta x^h_n )= f(x,a) \Delta t^h(x,a ) + o(\Delta t^h(x,a)),\\
& E_{x,n}^{h,a} \left(\left[\Delta x^h_n -  E_{x,n}^{h,a}\Delta x^h_n
\right]\left[\Delta x^h_n -  E_{x,n}^{h,a}\Delta x^h_n
\right]'\right)\\
&= g(x) g'(x)\Delta t^h(x,a)+o(\Delta t^h(x,a)),\\
& \sup_{n,\omega}\norm{\Delta x_n^h } _2 \xrightarrow{h} 0,
\end{align*}  
where $\Delta x^h_n = x_{n+1}^h- x_n^h$ is the difference and
$\Delta t^h(x,a)$ is an appropriate interpolation interval for
$x\in X$ and $a\in U$.  As $h\rightarrow 0$, the continuous
interpolations $x^h(\cdot), u^h(\cdot)$ of $\{x^h_n, n \in \bbZ \}$
and $\{u_n^h, n \in \bbZ\}$ under the interpolation times
$\{ t_n^h, n \in \bbZ \}$ computed from the interpolation intervals
$\Delta t_n^h= \Delta t^h(x_n^h ,u_n^h)$, $n \in \bbZ$, converge in distribution to
$x(\cdot), u(\cdot)$ which solve the \ac{sde} in
\eqref{sde}.\end{theorem}

% The local consistency condition requires that the Markov chain
% approximation has the ``local property'' (mean and variance) of the
% diffusion process \eqref{sde} with $x(t)=x$ and $u(t)=a$ on the
% interval $[t,t+ \Delta t^h(x,a)]$.

Given a spatial step $h$, under the local consistency condition we
construct the \ac{mdp} $M^h$ over the discrete state space $X^h$ by
computing the transition probability function $P^h$ from the
parameters of the \ac{sde} (see \cite{kushner2001numerical} for the
details). If the diffusion matrix $g(x)g(x)^T$ is diagonal, then the
transition probabilities are:
$ P^h(x,a, x\pm h_ie_i)= \Delta t(x,a) \cdot \left[\frac{
    (g(x)g'(x))_{ii}}{2h_i^2}+\frac{ f_i^{\pm}(x,a)}{h_i}\right], $
and
$ P^h(x,a,x) = 1- \Delta t(x,a) \cdot \sum_{i=1}^n\left[\frac{
    (g(x)g'(x))_{ii}}{h_i^2}+\frac{ \abs{f_i(x,a)}}{h_i}\right], $
where $e_i$ is the unit vector in the $i$-th direction and
$f_i^{\pm}(x,a)=\max(\pm f_i(x,a),0)$.

% In order for the local consistency condition to be satisfied, the
% interpolation interval needs to satisfy
% \begin{equation}
% \label{eq:constraintdelta}\Delta t^h(x,a)\le \frac{1}{\sum_{i=1}^n\left[\frac{ (g(x)
%       g'(x))_{ii}}{h_i^2} +\frac{\abs{f_i(x,a)}}{h_i}\right] } .\end{equation}
% In addition, we choose $\delta$ such that
% $\frac{\calV_c}{\delta} \in \bbZ$ for all $c\in C$. Note that in the
% definition of clock constraints, we only compare the clock values to
% integers. Choosing $\delta$ in such a way, it is ensured that at the
% interpolation time, we only encounter clock vectors in
% $\calV^\delta$.

\subsection{Optimal planning with the discrete approximation}
In this section, we construct a product \ac{mdp} from a discrete
approximation of the original system and the timed automaton
expressing the system specification.  Then, an optimal planning
problem is formulated in a product \ac{mdp} for computing a
near-optimal policy for the \ac{sde} in \eqref{sde} with respect to
the probability of satisfying the \ac{mitl} specification in the
point-based semantics.  % The product
% \ac{mdp} is computed from a discrete approximation of the original
% system and the timed automaton. We show that if a sample path of the
% original stochastic system under a control policy satisfies the
% specification in the point-based semantics for some sampling interval
% $\delta$, then a discrete path over the state space in the product
% \ac{mdp} induced from that sample path hits a particular set of
% states. 

% Then, we
% formulate an optimal planning problem in the product \ac{mdp} for
% computing a near-optimal policy with respect to the probability of
% satisfying the \ac{mitl} specification in the point-based semantics.
% A near-optimal policy for Problem~\ref{problem2} can then be
% obtained by solving the optimal planning problem in the product
% \ac{mdp}.

% The solution to this optimal
% planning problem converges to a solution to Problem~\ref{problem2}
% as the discretization parameter $h$ approaches $0$.

Given the timing constraints in \ac{mitl}, we consider an explicit
approximation method that discretizes both the continuous state space
and time. Particularly, instead of computing potentially varying
interpolation intervals, we choose a constant interpolation interval
$\delta$, referred to as the \emph{time step}. For the local
consistency condition to hold, it is required that for a given
$h\in \bbR^n$, 
\begin{equation}
\label{eq:constraintdelta}\delta \le \frac{1}{\sum_{i=1}^n\left[\frac{ (g(x)
      g'(x))_{ii}}{h_i^2} +\frac{\abs{f_i(x,a)}}{h_i}\right] },
\forall x\in X, \forall a\in U.\end{equation}
 % Hence, explicit approximation means that the discretization
% parameter for the time is also the constant interpolation interval.
Furthermore, $\delta$ is used as the parameter to discretize the clock
vector space $\calV$. Let
$\calV_i^\delta = \{k \delta \mid 0\le k \le \lceil
\frac{\max(\calV_i)}{\delta} \rceil\}$
be the discretized space for the range $\calV_i$ of clock $c_i \in
C$.
The discretized clock vector space is
$\calV^\delta = \Pi_{i=1,\ldots,M} \calV_i^\delta$.  Since both $X$
and $\calV$ are bounded, sets $X^h$ and $\calV^\delta$ are both
finite.  The method is ``explicit'' given the fact that the advance of
clock values are explicit: At each step $n$, if the clock is not reset
to $0$, then its value is increased by the interpolation interval
$\delta$. % In \cite{kushner2001numerical} an ``implicit''
% approximation method is introduced for finite-time horizon optimal
% control problems in which time is treated as a state variable but the
% interpolation interval may be different from the discretization
% parameter for time. The reader is referred to
% \cite{kushner2001numerical} for more information and comparison of
% these two methods.

Let $d=(h,\delta)$ denote a tuple of spatial and time steps. Next, we
construct a product \ac{mdp}
$\calM^d = M^h \times \calA_\varphi = \langle S^d, U, P^d, s_0
\rangle$
where $S^d = X^h \times Q\times \calV^\delta $ is the discrete product
state space, $U$ is the input space,
$P^d: S^d\times U \times S^d\rightarrow [0,1] $ is the transition
probability function, defined as follows. Let $s= (x,q,v)$ and
$s'=(x',q',v')$. For any $a\in U$, $P^d(s,a,s')= P^h(x,a,x')$ if and
only if $(q,v)\xrightarrow{\delta, L(x')} (q',v')$. Otherwise
$P^d(s,a,s')=0$. The initial state is $s_0 =(x_0,q_0,\bm 0 )$ with
$(\init, \bm 0) \xrightarrow{0,L(x_0)} (q_0, \bm 0)$.

\begin{assumption}
\label{assume:equivalent}
There exists a spatial step $h\in \bbR^n$ and a choice of
representative points from $X$ such that for all $x\in X^h$ and all
$x'\in [x]$, $L(x')=L(x)$. 
\end{assumption}
\begin{lemma}
  \label{lma1}Under Assumption~\ref{assume:equivalent}, given
  $x(\cdot), u (\cdot)$ that solve the SDE in \eqref{sde} and a
  discretization $X^h$ of the state space, we construct a discrete
  chain $\{S_n, n \in \bbZ\}$ as follows:
  $S_0 = s_0 =(x_0,q_0, \bm 0)$ with $x(0)\in [x_0]$ and
  $(\init, \bm 0) \xrightarrow{0,L(x_0)} (q_0, \bm 0)$; for all
  $n \in \nat$, $S_n = (x^h_n, q_n, v_n)$ where $x_n^h \in X^h$ is the
  representative point to which $x(n\delta)$ belongs, i.e.,
  $x(n\delta)\in [x_n^h]$, and
  $(q_n, v_n)\xrightarrow{\delta, L(x_{n+1}^h)} (q_{n+1}, v_{n+1})$.
  The following two statements hold.
\begin{enumerate}
\item For all $n \in\bbZ$, the range of the random variable $S_n$ is
  $S^d$.
\item The probability of a continuous sample path in
  $\{x(t), t\ge 0\}$ satisfying $\varphi$ in the point-based semantics
  under the sampling interval $\delta$ equals the probability of a
  discrete sample path in the chain $\{S_n, n \in \bbZ\}$ hitting the
  set $G^d = X^h\times F\times \calV^\delta$. That is,
\begin{multline*}
  P_{x(0)} \left( [L(x(\cdot))]^\delta\models \varphi  \right) =\\
  P_{s_0}(S_k \in G^d\text{ and } \forall  j<k, S_j \notin G^d).
\end{multline*}
% where $\Pr \left( [L(x(\cdot))]^\delta\models \varphi \right) $ is
% the probability of a path in the process $x(t), t\ge 0$ that
% satisfies $\varphi$ in the point-based semantics under the sampling
% interval $\delta$.
\end{enumerate}
\end{lemma}

\begin{proof}
  To show the first
  statement, % it is noted that in the set $\Phi_C$ of
  % clock constraints, the clock values are compared to non-negative
  % integers.
  initially, $v(0)=\bm 0 $ is a vector of zeros, which is in
  $\calV^\delta$. Suppose that at the $n$-th sampling step
  $v(n\delta) \in V^\delta$, at the next sampling step, for any clock
  $c_i \in C$, either the value of $c_i$ is increased by $\delta$ or
  it is reset to $0$ depending on the current state in the automaton,
  the clock vector and the current labeling of state in $X$.  If the
  value of $c_i$ is reset to $0$,
  $v((n+1)\delta)[i]=0 \in \calV_i^\delta$. Otherwise,
  $v((n+1)\delta)[i]=v(n\delta)[i]+\delta \in \calV_i^\delta $.  By
  induction, all possible clock vectors we can encounter at the
  sampling times are in the set $\calV^\delta$. Thus, the range of a
  random variable $S_n$ for any $n \in \bbZ$ is $S^d$, which is a
  subset of $X^h\times Q\times \calV$.

  Let $x(\cdot,\omega)$ with $\omega \in \Omega$ be a sample path of
  the process $x(\cdot)$ and $\mathsf{p}= s_0s_1\ldots $ be the
  corresponding sample path of the chain $\{S_n, n \in \bbZ\}$ given
  the construction method above. Remind that $x(\cdot,\omega)$
  satisfies $\varphi$ in the point-based semantics under the sampling
  interval $\delta$ if and only if the timed word $\calT(b^\delta) $,
  where $b(\cdot) = L(x(\cdot, w))$, is accepted in $\calA_\varphi$.
  Since $x(i\delta,\omega) \in [\pi_1(s_i)] $ for all $ i \ge 0$, let
  $\tau = 0\ \delta\ 2\delta \ldots$, we have that the timed word
  $\calT(b^\delta) = (\tau, L(\pi_1(\mathsf{p})))$ by
  Assumption~\ref{assume:equivalent} . Let the run on the timed word
  $(\tau, L(\pi_1(\mathsf{p})))$ be $\rho $ such that
  $\rho= (\init, \bm 0)\xrightarrow{0, L(\pi_1(s_0))}(q_0,
  v_0)\xrightarrow{\delta, L(\pi_1(s_1))}(q_1,v_1)\ldots$.
  By construction, it holds that $q_i = \pi_2(s_i)$ and
  $v_i = \pi_3(s_i)$ for all $i \in \bbZ$. By definition of the
  acceptance condition in $\calA_\varphi$,
  $(\tau, L(\pi_1(\mathsf{p})))$ is accepted in $\calA_\varphi$ if and
  only if $\Occ(\rho)\cap F \ne \emptyset$, which is equivalent to say
  that for some $k \in \bbZ$, $s_k \in X^h\times F\times \calV$ and
  for all $j<k$, $s_j \notin X^h\times F\times \calV$. Since in the
  first statement we have shown that the range of $S_n$ for all
  $n \in \bbZ$ is $S^d$,
  $s_k \in X^h\times F\times \calV^\delta = G^d$ and the proof for the
  second statement is complete.
\end{proof}
% \begin{lemma}Assuming \ref{assume:equivalent}, let
%   $\rho= s_0s_1\ldots s_n\in (S^d)^\ast$ be a path in the product
%   \ac{mdp} $\calM^d$ and $\delta$ be the sampling interval. For a
%   given finite-time state trajectory $x(t,w), 0\le t\le n\delta$ whose
%   sampled trajectory $\rho' = x_0x_1\ldots x_n$ with $x_i= x(i\delta)$
%   satisfies $x_i \in [\pi_1(s_i)]$, for all $i \in \bbZ$,
%   $0\le i \le n$, the timed behavior $L(x(t))$ of $x(t)$ satisfies the
%   \ac{mitl} formula $\varphi$ in the point-based semantics under the
%   sampling interval $\delta$ if and only if $s_n \in G^d$. 
% \end{lemma}
% \begin{proof}
%   By Assumption~\ref{assume:equivalent}, for any state trajectory
%   $x(t), 0\le t \le n\delta$, if for all $ i \ge 0$,
%   $x_i =x(i\delta) \in [\pi_1(s_i)] $, then
%   $L(\rho' )= L(\pi_1(\rho))$. Let $ \tau = 0 \delta \ldots
%   n\delta$.
%   By construction of the product \ac{mdp},
%   $(\tau,L(\rho')) = (\tau, L(\pi_1(\rho)))$ is accepted by the timed
%   automaton $\calA_\varphi$ if and only if $s_n\in G^d$. % Also
%   % note that $L(\rho')$ is the sampled timed behavior of $L(x(t))$, the
%   % proof is completed.
% \end{proof}

% Intuitively, the above Lemma shows that if a continuous state
% trajectory $x(t), t \ge 0 $ satisfies the \ac{mitl} formula $\varphi$
% in the point-based semantics under the sampling interval
% $\delta$. Then 
Lemma~\ref{lma1} characterizes the probability of satisfying the
specification in point-based semantics under the sampling interval
$\delta$ with the probability of reaching a set $G^d$ in the product
\ac{mdp} $\calM$. Given the objective of maximizing the probability of
reaching a set in an \ac{mdp}, there exists a memoryless
and deterministic policy such that by following this policy, from any
state, the probability of reaching the set is maximized
\cite{Baier2008}.

We introduce a state $\sink$ into the product \ac{mdp} $\calM^d$ and
modify $P^d$ such that for all $s \in G^d$ and all $a\in U$,
$ P^d(s, a,\sink ) = 1$, and for all $a\in U$,
$P^d(\sink, a, \sink)=1$, while the other transition probabilities
remain unchanged. The product \ac{mdp} $\calM^d$ with the augmented
state set and the modified transition probability function is denoted
$\hat \calM^d$.  The reward function
$R: S^d \cup\{\sink\} \rightarrow \bbR $ is defined by
$R(s )=1_{G^d}(s) $.  Let $u: S^d \cup\{\sink\} \rightarrow U$ be a
memoryless and deterministic policy in $\hat \calM^d$ and $\Pi^d$ the
set of all such policies in $\hat \calM^d$.  The value function of
policy $u$ is
\[ W^d(s, u) = E_{s}^u \left[ \sum_{i=0}^\infty R(S_i) \right],\]
where $\{S_n, n\in \bbZ\}$ is the Markov chain induced from
$\hat \calM^d$ with policy $u$.  Thus, the optimal value function
$V^d(s) = \max_{u \in \Pi^d} W^d(s, u)$, and the dynamic programming
equation is obtained: For $s\in S^d$,
\begin{align*}
\begin{split}
\label{eq:optfun}
&V^d(s)=  R(s)+ \max_{a\in U}\left[ \sum_{s'\in S^d \cup\{\sink\}} P^d(s, a,s') \cdot
  V^d(s') \right],\\
&\text{ and } V^d(\sink)=0.
\end{split}
\end{align*}

% Let $u:S^d\rightarrow U$ be the optimal policy in the product
% \ac{mdp}. A near optimal policy to Problem~\ref{problem2} is computed
% by 
% \paragraph*{\bf Implementation of the optimal policy in
% $\hat\calM^d$}

Given the optimal policy
$\hat{u}^\ast: S^d \cup\{\sink\} \rightarrow U$ that achieves the
maximum value of $V^d(s)$ for all $s\in S^d$ in the modified product
\ac{mdp} $\hat \calM^d$, we derive a policy $u^\ast:S^d\rightarrow U$
by letting $u^\ast(s)= \hat{u}^\ast(s)$. By the definition of reward
function and the modified product \ac{mdp}, policy $u^\ast$ maximizes
the probability of hitting the set $G^d$ in $\calM^d$.

A policy $u:S^d\rightarrow U$ is implemented in the original system in
\eqref{sde} in the following manner. The initial product state is
$(x(0), q_0, \bm 0)$ with
$(\init, \bm 0)\xrightarrow{0, L(x(0))} (q_0,\bm 0)$. At each sampling
time $n\delta, n\in \bbZ$, let the current product state be
$(x, q,v) \in S$. We compute $(x^h, q, v) \in S $ such that
$x\in [x^h]$. Note that at the sampling time the clock vector is
always in $\calV^\delta$ by Lemma~\ref{lma1} and thus
$(x^h, q,v)\in S^d$, for which $u$ is defined. Then, we apply a
constant input $u((x^h,q, v))$ during the time interval
$[n\delta, (n+1)\delta)$. At the next sampling time $(n+1)\delta$,
according to the current state $x'$, we compute the state in
$\calA_\varphi$ and the clock vector such that
$(q,v)\xrightarrow{\delta, L(x')} (q',v')$.  Hence, the new product
state is $(x',q',v')$ and a constant control input for the interval
$[(n+1)\delta, (n+2)\delta)$ is obtained in the way we just described.
%  In the \ac{mdp} $\hat \calM^d$ with
% The dynamic programming equation indicates a memoryless, deterministic
% policy in the product \ac{mdp} $\hat{\calM}^d$ can attain the maximal
% probability of reaching a state in $G^d$ \cite{sutton98a}.

\begin{remark}
  When the input space $U$ for the product \ac{mdp} is bounded, in the
  numerical method for the reward maximization problem in the product
  \ac{mdp}, in general we also discretize the input space $U$ with
  some discretization parameter $\epsilon$. Let $U^\epsilon$ be the
  discretized input space. Given the optimal policy $u^\ast$ for the
  product \ac{mdp} and the optimal policy $u^{\epsilon, \ast}$ in the
  product \ac{mdp} with the input space $U^\epsilon$, one can derive
  the bound on $\abs{W^d(s_0, u^\ast) - W^d(s_0, u^{\epsilon,\ast})}$
  as a function of $\epsilon$, which converges to $0$ as
  $\epsilon \rightarrow 0$
  \cite{fleming2006controlled,kushner2001numerical}.  Thus, with both
  discretized state and input space, the implemented policy is
  near-optimal for the SDE in \eqref{sde} with respect to the
  probability of satisfying the \ac{mitl} specification in the
  point-based semantics.
\end{remark}

 \subsection{Proof of convergence}
% \label{subsec:convergence}
 Based on Theorem~\ref{thm:stateconverge}, we show that the optimal
 policy synthesized in the product \ac{mdp} converges to the
 optimal policy that achieves the maximal probability of satisfying
 the \ac{mitl} specification in the dense-time semantics as the
 discretization in both state space and time space get finer.
\begin{theorem}
  Given a discretization parameter $d= (h,\delta)$ where $\delta$
  satisfies the local consistency condition in
  \eqref{eq:constraintdelta} with respect to the spatial step $h$, it
  holds that
  \[ \lim_{h \rightarrow 0}V^d(s_0) = V(s(0)),
\]
where $s_0 = (x_0, q_0, \bm 0)$, $s(0) = (x (0), q_0, \bm 0)$ and
$x(0)\in [x_0]$.
\end{theorem}

\begin{proof}
  First, it is noted by the local consistency condition and the
  constraint on $\delta$ in \eqref{eq:constraintdelta}, $\delta$ is a
  decreasing function of $h$ and when $h \rightarrow 0$,
  $\delta \rightarrow 0$.

  For a given $d=(h,\delta)$ where $\delta$ satisfies the constraint
  in \eqref{eq:constraintdelta} with respect to $h$, let
  $u^d: S^d \rightarrow U$ be a policy in the product \ac{mdp}
  $\calM^d$ and $\{S^d_n, n \in \bbZ\} $ be the induced Markov chain.
  According to Theorem~\ref{thm:stateconverge}, when
  $h \rightarrow 0$, $\delta \rightarrow 0$ and the continuous
  interpolations of $\{\pi_1(S^d_n), n \in \bbZ\}$ and
  $\{u^d_n = u(S^d_n), n \in \bbZ\}$ converge in distribution to
  $x(\cdot)$ and $u(\cdot)$ that solve the SDE in \eqref{sde}. By the
  determinism in the transition function of the timed automaton and
  the labeling function, as $h\rightarrow 0$, $\{S^d_n, n \in \bbZ \}$
  also converges in distribution to $\{s(t), t\ge 0\}$, which is the
  product stochastic process derived from $\{x(t), t \ge 0\}$.
  According to the definition of reward functions $r$ and $R$, since
  $h \rightarrow 0$, we have $\delta \rightarrow 0$,
  $ x_0 \rightarrow x(0)$ and $W^d(s_0, \{u^d_n, n\in \bbZ\}) $
  converges to $W(s(0), \{u(t), t\ge 0\} )$.

  Now given the optimal control policy $u^{d,\ast}: S^d\rightarrow U$
  obtained for the product \ac{mdp} $\calM^d$, let
  $\{S_n, n \in \bbZ\}$ be the Markov chain induced by $u^{d,\ast}$ in
  $\calM^d$. We have that
  $\lim_{h \rightarrow 0}W^d (s_0, \{u^{d,\ast}(S_n), n\in \bbZ\}) =
  \lim_{h \rightarrow 0}\sup V^d(s_0)\le V(s(0))$
  by the optimality of the value function $V(s(0))$.  On the other
  hand, let $u^\ast = \arg\sup_{u\in \Pi} W(s(0),u)$ be the optimal
  control policy in the continuous-time stochastic system.  We
  construct a policy $\{u^\ast(n\delta), n \in \bbZ \}$ for the
  product \ac{mdp} $ \calM^d$ such that the action $u^\ast(n\delta) $
  is taken at the step $n$.  We have
  $V^d(s_0) \ge W^d(s_0, \{ u^\ast(n\delta), n \in \bbZ \})$ by the
  optimality of $V^d(s_0)$.  Since
  $\lim_{h \rightarrow 0} W^d(s_0, \{u^\ast(n\delta), n \in \bbZ \}) =
  W(s(0), \{u^\ast(t), t\ge 0 \})= V(s(0))$,
  it is inferred that
  $\lim_{h \rightarrow 0} \inf V^d(s_0) \ge V(s(0))$. Therefore,
  $\lim_{h \rightarrow 0 } V^d(s_0) = V(s(0))$.
\end{proof}
\section{Example}

\begin{figure*}[!t]
\centering
    \begin{subfigure}[b]{0.45\textwidth}
\centering
\includegraphics[width=\textwidth]{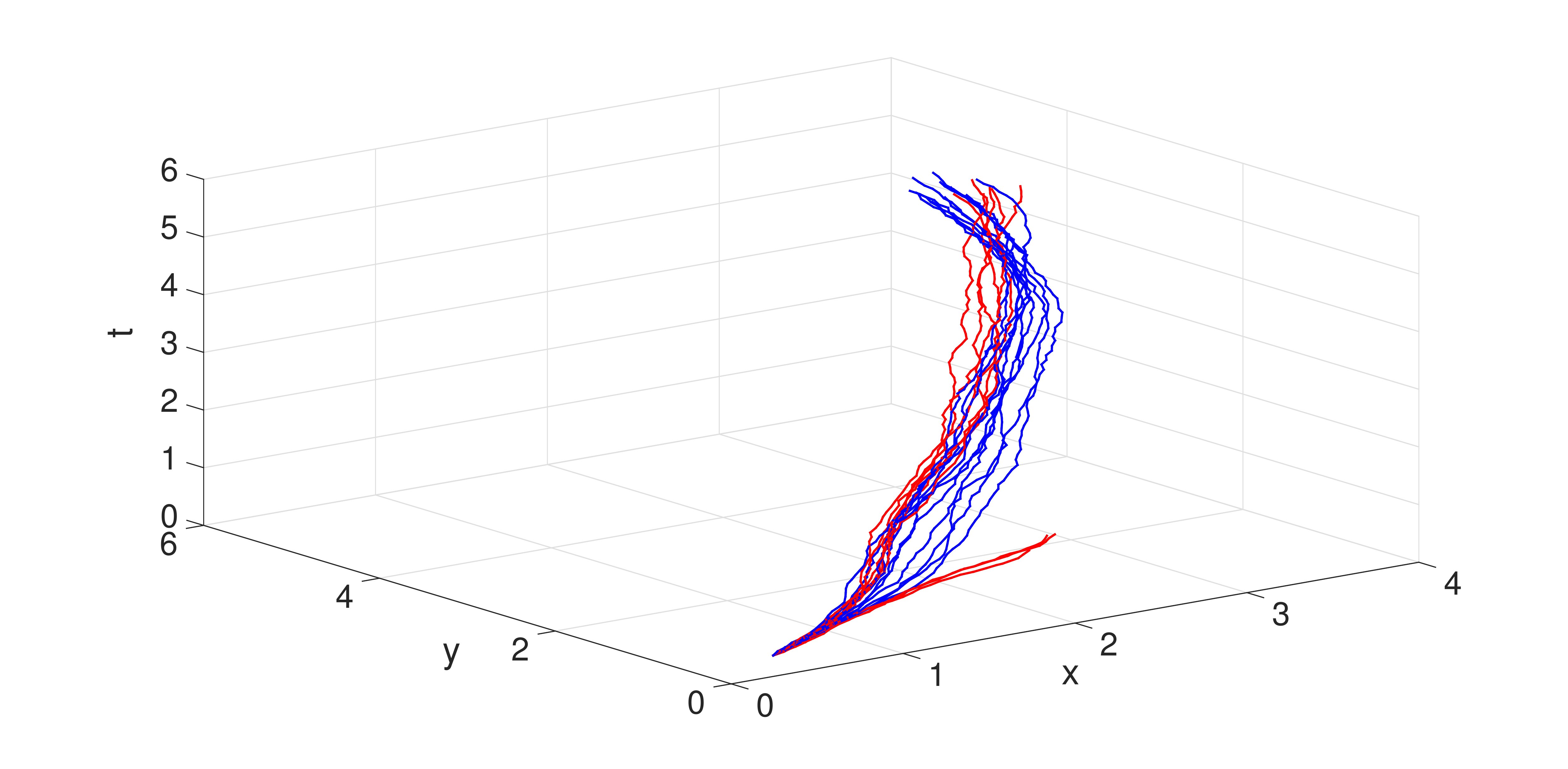}
\caption{3D-view}
\label{fig:3d-view}
\end{subfigure}
    \begin{subfigure}[b]{0.45\textwidth}
\centering
\includegraphics[width=\textwidth]{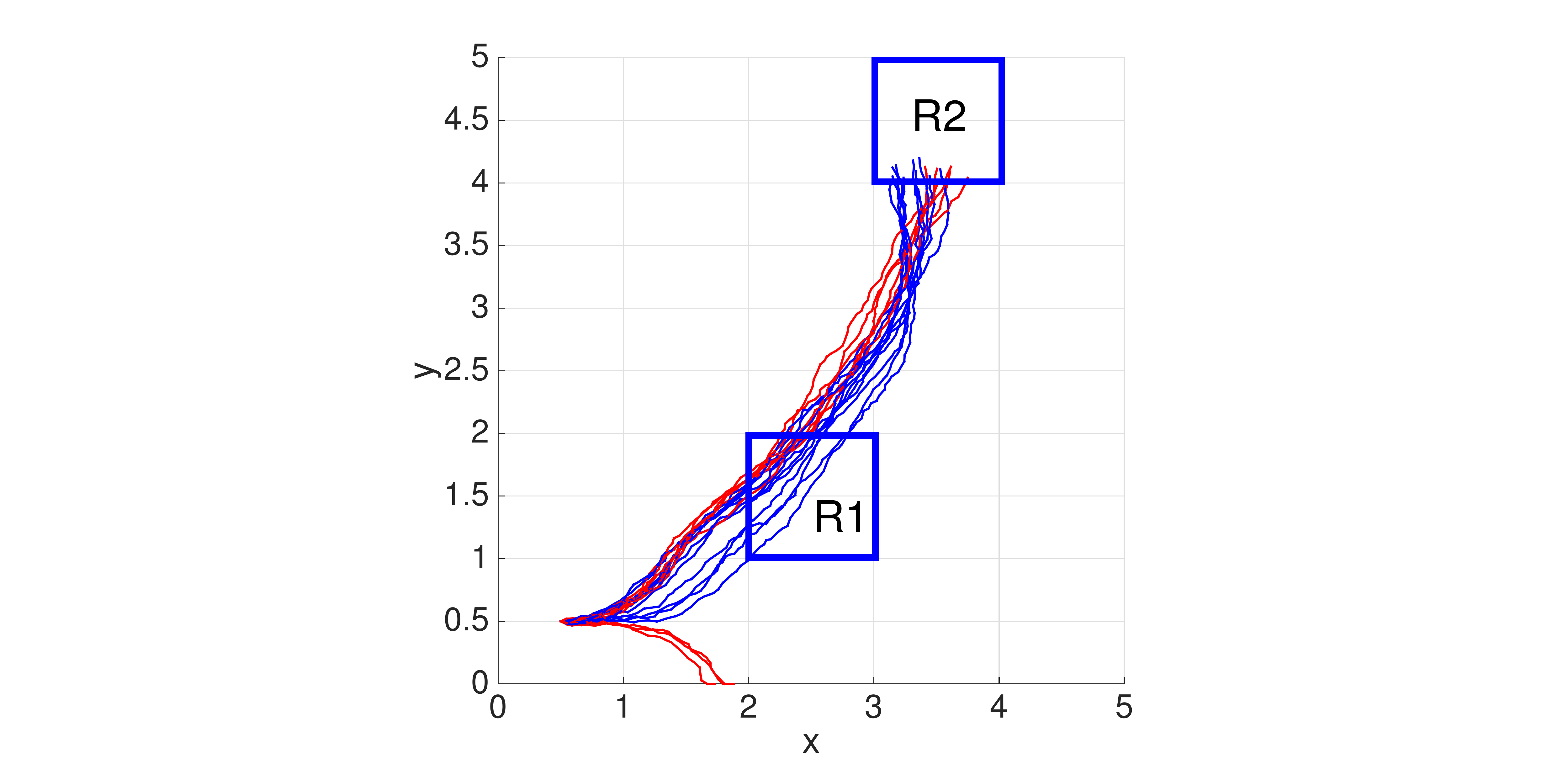}
\caption{x-y view}
\label{fig:xy-view}
\end{subfigure}
    \begin{subfigure}[b]{0.45\textwidth}
\centering
\includegraphics[width=\textwidth]{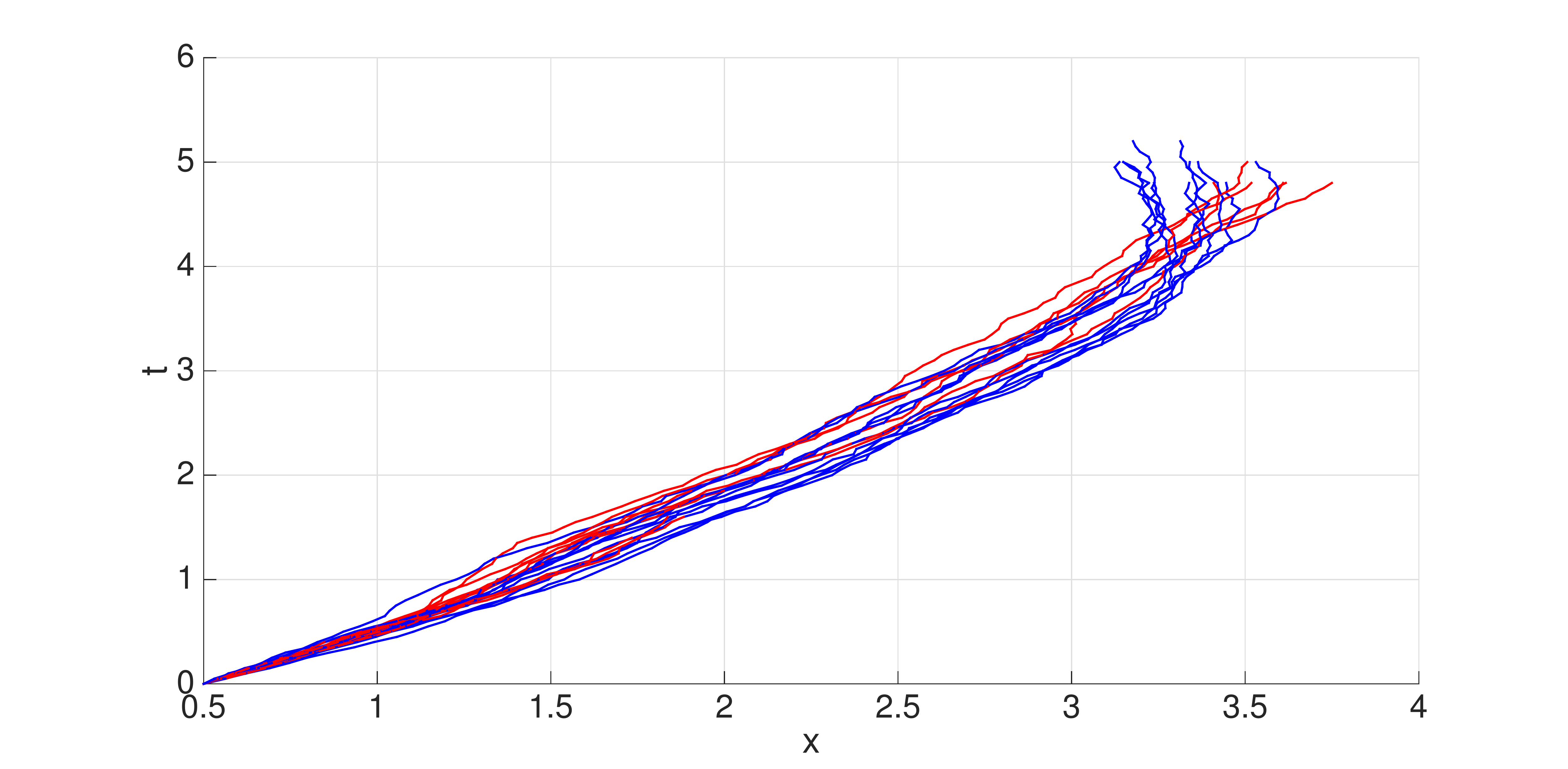}
\caption{x-t view}
\label{fig:xt-view}
\end{subfigure}
    \begin{subfigure}[b]{0.45\textwidth}
\centering
\includegraphics[width=\textwidth]{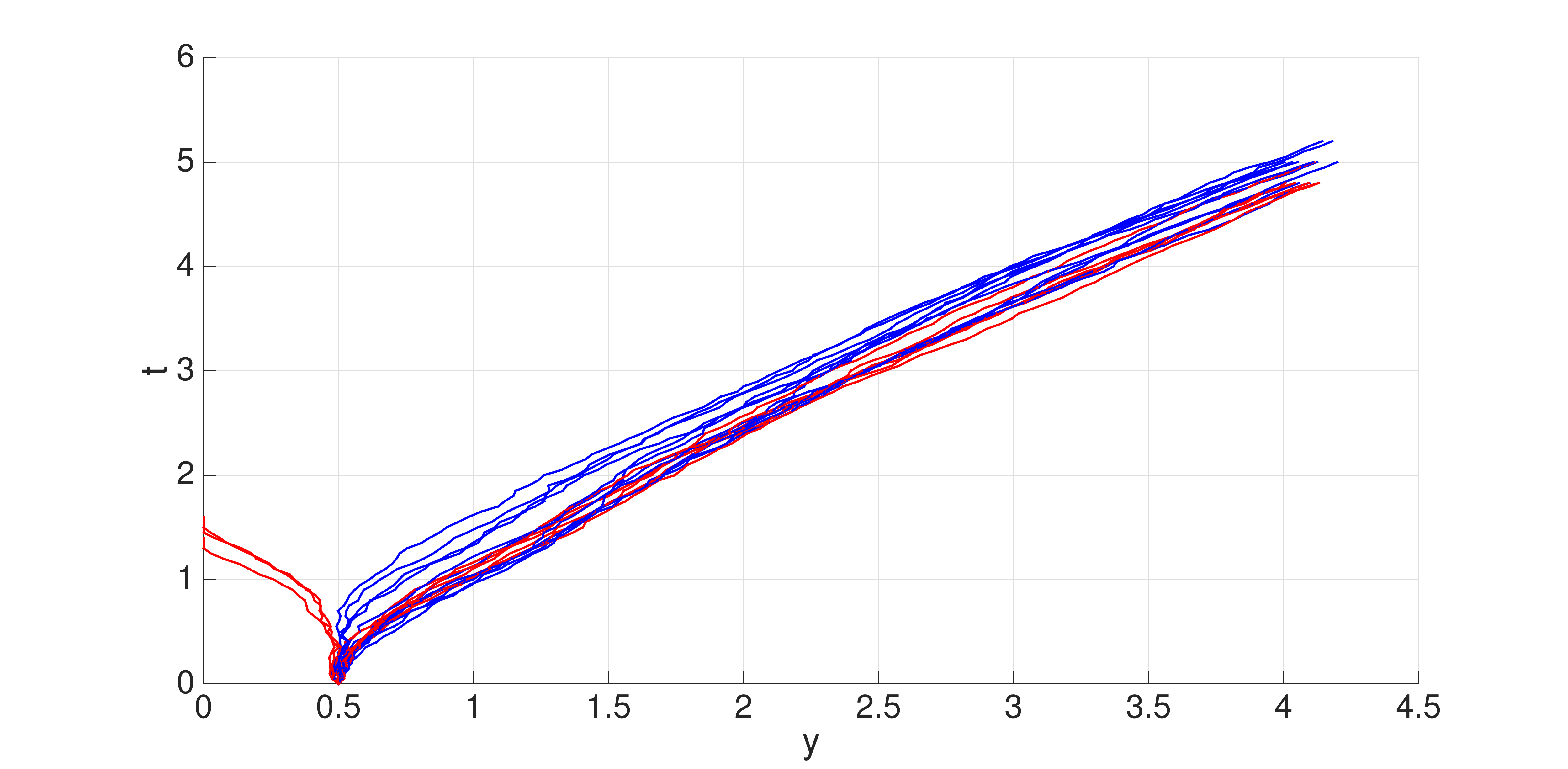}
\caption{y-t view}
\label{fig:yt-view}
\end{subfigure}
\vspace{-2ex}
\begin{subfigure}[b]{0.5\textwidth}
\centering
\includegraphics[width=\textwidth]{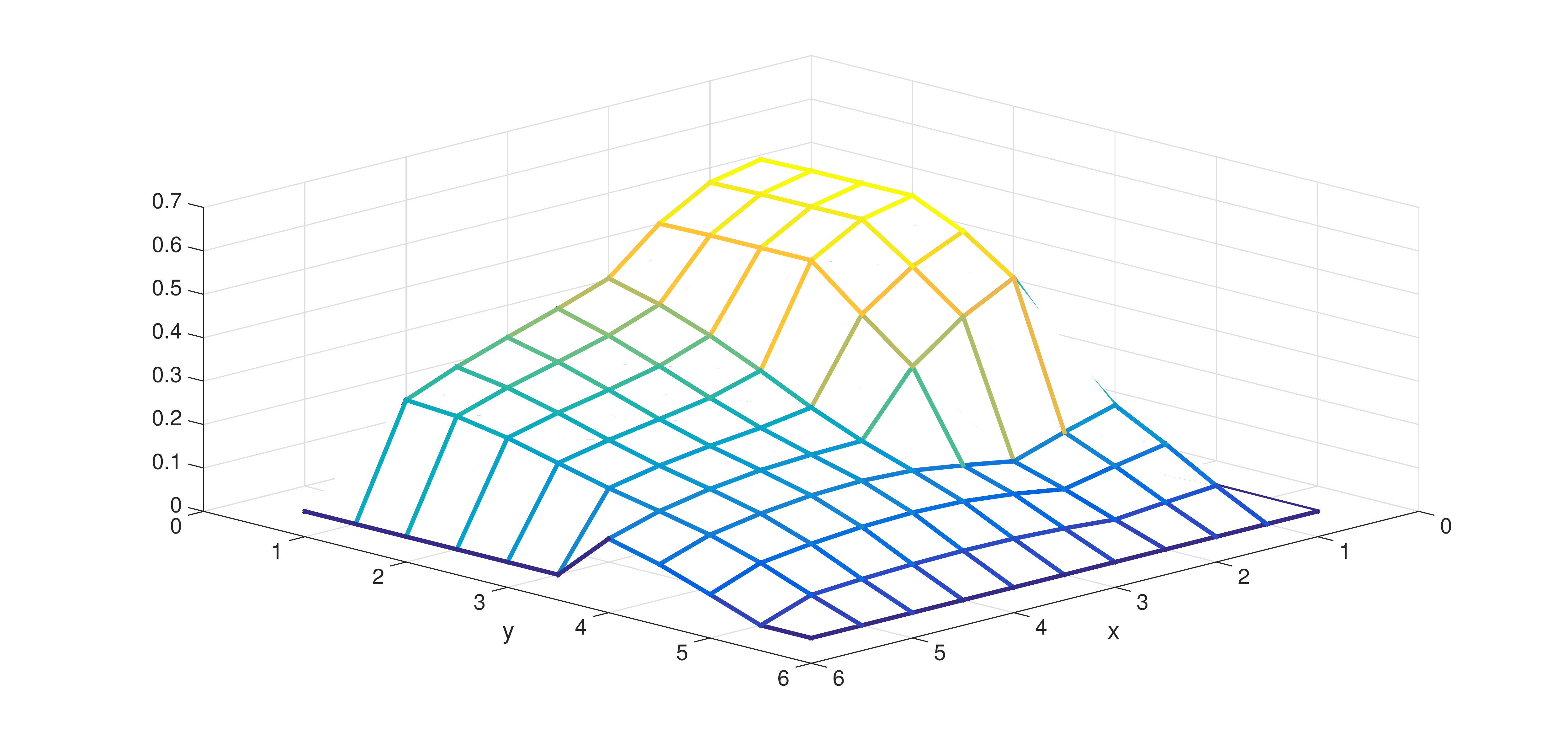}
\caption{Value function}
\label{fig:valuefunction}
\end{subfigure}
\label{fig:simulation}
\caption{(a) -- (d) Total 20 sample paths for the robot starting at
  $x_0=(0.5,0.5,0)^T$ with 3D view, x-y plane view, x-t plane view and
  y-t plane view. The sample path stops whenever the specification is
  satisfied, or it does not meet the specification due to time
  constraints or hitting the walls.  Most of the sample paths that
  fail to satisfy the specification in the point-based semantics reach
  the region $R_2$ prior to the $3$rd time units after visiting
  $R_1$. (e) The value function for the robot with the initial state
  $\bm x_0=(x,y,0)^T \in X^h$, the initial state in the specification
  timed automaton and the initial clock vector $\bm 0$.}
\end{figure*}
This section illustrates the method using a motion planning example
for a robot modeled as a stochastic Dubin's car. The dynamics of the
system are described by the \ac{sde}
\[
\underbrace{\begin{bmatrix} dx(t) \\
dy (t)\\
d \theta(t)
\end{bmatrix} }_{d \bm x(t)} = \underbrace{\begin{bmatrix}
  v(t) \cos \theta(t) \\
  v(t) \sin \theta(t) \\
  u(t)
\end{bmatrix}dt}_{f( \bm x(t),u(t))dt}
+ g(\bm x(t))dw,
\] 
where $\bm x =(x,y, \theta)$ is the coordinate and heading angle of
the robot, $v$ is the linear velocity and $u \in U=[-1,1]$ is the
angular velocity input. In this example, $v=1$ is fixed and
$g(\bm x(t))=0.5 I_3$, and $w(\cdot)$ is a 3-dimensional Wiener
process on the probability space $(\Omega, \calF, P)$.

The workspace of the robot is depicted in Figure~\ref{fig:xy-view},
with two regions $R_1$ and $R_2$ of importance. The workspace is
constrained by the walls
$\{(x,y)\mid x\in \{0,5\}, 0\le y\le 5\} \cup \{(x,y)\mid 0\le x \le
5, y \in \{0,5\}\}$.

The objective of the robot is to maximize the probability of visiting
region $R_1$ within the first $5$ time units and after visiting $R_1$,
reaching $R_2$ between the $3$rd and $5$th time units, while avoiding
hitting the walls. We define atomic propositions $ R_i$, $i=1,2$,
which evaluates true when the robot is in region $R_i$. An atomic
proposition $HitWall$ evaluates true if the robot hits the surrounding
walls.  The \ac{mitl} formula describing the specification is
$\varphi=\lozenge_{[0,5]}( (R_1 \land \neg HitWall)\land
\lozenge_{[3,5]}( R_2\land \neg HitWall))$.
Given an initial state $\bm x_0$, we want to find an optimal policy
that maximizes the probability of $\varphi$ being satisfied. We select
a spatial step $h=(0.5,0.5,\pi/4)^T$ to obtain a uniform
discretization of the state space $X$. Given the choice of $h$, the
time step $\delta$ is chosen to be $0.2$ time units for the local
consistency condition to hold for all state and control input
pairs. The number of states in the \ac{mdp} $M^h$ is $1089$ and the
number of product states in the modified product \ac{mdp}
$\hat{\calM^d}$ is $58809$ (after trimming unreachable states). Remind
that the value iteration is polynomial in the size of the \ac{mdp}
$\hat\calM^d$.  The implementation are in
MATLAB\textsuperscript{\textregistered} on a desktop with Intel(R)
Core(TM) processor and 16 GB of memory. The computation of the product
\ac{mdp} takes $18$ minutes and the value iteration converges after
$50$ iterations with a pre-specified error tolerance of $0.01$. Each
iteration takes about $6$ minutes. In the value iteration we also
approximate the input space $U$ with a finite set $U^\epsilon$ where
$\epsilon=0.2$ is the discretization parameter for the input space.

Since the product state space of the example is $5$-dimensional, we
select to plot the optimal value $V^h$ for the states with the initial
heading angle $\theta = 0$, the initial state of the timed automaton
and initial clock vector $\bm 0$ in
Figure~\ref{fig:valuefunction}. Figures~\ref{fig:3d-view},
\ref{fig:xy-view}, \ref{fig:xt-view}, and \ref{fig:yt-view} show the
sample paths starting from $\bm{x}_0=(0.5,0.5,0)$ for a time interval
$[0,6]$ from different perspectives. The optimal value $V^d(s)$ with
$s= ((0.5,0.5,0),\init,\bm 0)$ is $0.54$, which is the approximately
maximal probability for satisfying $\varphi$ in the point-based
semantics under the sampling interval $0.2$ in the system with initial
state $\bm x(0)=(0.5,0.5,0)$. In simulation, there are $11$ paths
(marked in blue) out of $20$ sample paths that satisfy the
specification in the point-based semantics. % The red paths violates
% the \ac{mitl} formula in the point-based semantics.
% We use the Euler-Maruyama method to numerically approximate the
% controlled stochastic process.

% \begin{figure}[ht]
% \centering
% \includegraphics[width=0.45\textwidth]{figures/workspace}
% \label{fig:workspace}
% \caption{The workspace of the robot.}
% \end{figure}
% \begin{figure}[h!]
% \centering
% \includegraphics[width=0.5\textwidth]{figures/valuefunction}
% \label{fig:valuefunction}
% \caption{}
% \end{figure}

The drawback of the explicit approach is scalability. In order to
compute a control policy with a finer approximation, we need to reduce
the spatial step $h$ as well as the time step $\delta$ for the local
consistency condition to hold. The product state space becomes very
large for a fine discretization. For example, if $h$ is chosen to be
$(0.2,0.2, \pi/4)^T$, $\delta$ has to be chosen below $0.1$ time units
and for the simple example, the product \ac{mdp} has $608303$ states
after trimming. We did not carry out the computation for $V^h$ given
this finer discretization since it is very time consuming. We discuss
the limitation and possible solutions to deal with the issue of
scalability in Section~\ref{sec:conclude}. % Realizing
% the limitation, we will instead consider the implicit approximation
% method \cite{kushner2001numerical} in which the clock vector is
% treated as a state variable, and thus the discretization parameter for
% the clock vector space is not constrained by the interpolation
% interval, which can be very small given a finer grid for the local
% consistency condition to hold. This is part of ongoing research.

\section{Conclusions and future work}
\label{sec:conclude}
This paper proposes a numerical method based on the Markov chain
approximation method for stochastic optimal control with respect to a
subclass of quantitive metric temporal logic specifications. We show
that as the discretization gets finer, the optimal control policy in
the discrete abstract system with respect to satisfying the \ac{mitl}
specification in the point-based semantics converges to the optimal
policy in the original system with respect to the dense-time semantics
for satisfying the \ac{mitl} formula. The approach can be easily
extended to bounded-time MTL formulas including signal temporal logic
formulas. In the future work, we aim to investigate the error bounds
introduced by the proposed discrete approximation method. On the other
hand, since scalability is a critical issue in the explicit
approximation method, we will also investigate a solution approach
based on implicit approximation \cite{kushner2001numerical}. With
implicit approximation method, we can potentially reduce the size of
discrete abstract system by treating the clock vector as a state
variable, whose discretization parameters are pre-defined and
potentially different from the interpolation interval. Parallel
algorithms and distributed planning for large-scale \ac{mdp}s are also
considered to handle the issue of scalability.

\bibliographystyle{ieeetran}

\end{document}